\documentclass[11pt,english]{article}
\usepackage{lmodern}

\usepackage[T1]{fontenc}
\usepackage[utf8]{inputenc}
\usepackage{geometry}
\usepackage{color}
\usepackage{babel}
\usepackage{verbatim}
\usepackage{amsmath}
\usepackage{amsthm}
\usepackage{amssymb}

\usepackage[unicode=true,pdfusetitle, bookmarks=true,bookmarksnumbered=false,bookmarksopen=false, breaklinks=false,pdfborder={0 0 1},backref=false,colorlinks=true]{hyperref}
\hypersetup{pdfborderstyle=,citecolor=blue,filecolor=blue,urlcolor=blue,linkcolor=blue}

\theoremstyle{definition}

\newtheorem{thm}{Theorem}

\newtheorem{lem}{Lemma}

\newtheorem{cor}{Corolary}

\newtheorem{defn}{Definition}[section]

\newtheorem{prop}{Proposition}

\newtheorem{rem}{Remark}

\newcounter{mnotecount}[section]

\makeatother

\begin{document}
\title{Singularity theorems and the inclusion of torsion in affine theories of gravity}
\author{Paulo Luz$^{(1,2,3)}$ and Filipe C. Mena$^{(1,3)}$\\
\\
{\small{}$^{(1)}$Centro de Análise Matemática, Geometria e Sistemas Dinâmicos,}\\
{\small{}Instituto Superior Técnico, Universidade de Lisboa, Av. Rovisco Pais 1, 1049-001 Lisboa, Portugal} \\
{\small{}$^{(2)}$ Centro de Astrofísica e Gravitação - CENTRA, Departamento de Física,}\\
{\small{}Instituto Superior Técnico, Universidade de Lisboa, Av. Rovisco Pais 1, 1049-001 Lisboa, Portugal} \\
{\small{}$^{(3)}$Centro de Matemática, Universidade do Minho, 4710-057 Braga, Portugal} }

\maketitle
%%%%%%%%%%%%%%%%%%%%%%%%%%%%%%%%%

\begin{abstract}
We extend the scope of the Raychaudhuri-Komar singularity theorem of General Relativity to affine theories of gravity with and without torsion. We first generalize the existing focusing theorems using time-like and null congruences of curves which are hypersurface orthogonal, showing how the presence of torsion affects the formation of focal points in Lorentzian manifolds. Considering the energy conservation on a given affine gravity theory, we prove new singularity theorems for accelerated curves in the cases of Lorentzian manifolds containing perfect fluids or scalar field matter sources.
\end{abstract}

\section{Introduction}

The existence and genericity of space-time singularities has been an issue of great debate since the early years of General Relativity (GR), when they were thought to be a problem of some idealized and highly symmetric solutions of the Einstein field equations, and more realistic models were supposed to overcome such problems \cite{Senovilla_Garfinkle}.
This debate propelled Raychaudhuri \cite{Raychaudhuri_1955} and, independently, Komar \cite{Komar_1956}, to study the nature of singularities and their relation to the space-time symmetries. In his seminal work \cite{Raychaudhuri_1955}, studying the kinematics of a fluid in the context of GR, Raychaudhuri derived an equation for the evolution of the expansion of time-like geodesic congruences. Applying this equation to a model with a particular physical source, he provided the first evidence that the occurrence of singularities was not a direct consequence of the imposed symmetries  (see e.g. \cite{Senovilla_1998} for a review).

However, it was not until a decade later that the results derived by Penrose \cite{Penrose_1965}, and later extended by Hawking (see e.g. \cite{Haw_Ellis}), established that under more general physical conditions, the formation of space-time singularities was inevitable. Relying on the Raychaudhuri equation, and imposing some rather general constraints on the space-time geometry, it was possible to establish sufficient conditions for a congruence of geodesics to focus on a point, usually referred as a focal point of the congruence (also called caustic)\footnote{In this article, we will make no distinction between focal or conjugate points, see e.g. \cite{Senovilla_1998} for the precise definitions.}. Using the property that geodesics are extremal curves of the space-time interval, Penrose then related the formation of a focal point to the occurrence of a space-time singularity, in the sense that the geodesics of the congruence were inextensible beyond that point.

After the initial works of Hawking and Penrose, many extensions and reformulations to their singularity theorems have been made (see e.g. \cite{Senovilla_1998, Senovilla_Garfinkle} for a review in the context of space-times without torsion). In particular, since the 1970's there was a growing interest in studying whether the inclusion of space-time torsion could prevent the formation of singularities \cite{Trautman, Stewart, Kop}. This led to the formulation of singularity theorems for space-times with torsion \cite{Hehl1, Schmidt}. Indeed, using the incompleteness of time-like metric geodesics to probe the regularity of space-times in the Einstein-Cartan theory of gravity, Hehl \emph{et al}. \cite{Hehl2} found that the original Hawking-Penrose singularity theorems were equally valid in that theory, with the correction that the matter constraints had to be imposed on a modified stress-energy tensor. More recently, \cite{Torralba} related the occurrence of singularities with the existence of horizons, and formulated singularity theorems of Hawking-Penrose type for non-geodesic curves in particular Teleparallel and Einstein-Cartan theories of gravity.

In this paper, we consider the more general setting of affine theories of gravity with a non-symmetric affine connection, compatible with the metric. Our main questions are: What are the new geometric necessary conditions for the existence of focal points? How can we generalize the focusing theorems to any causal curves, not necessarily geodesics, in Lorentzian manifolds with torsion? In which cases can we prove singularity theorems from the focusing theorems, without imposing any further constrains on the manifold?

In order to answer those questions, we need to include general torsion terms in the structure equations of Lorentzian manifolds and carefully analyze the new terms. We will rely on the assumption that a hypersurface orthogonal congruence of curves exists. So, first, in Section \ref{frobeni}, we establish conditions for the existence of such congruences based on Frobenius theorem.  Then, in sections \ref{sec:Caustics_timelike} and \ref{sec:Caustics_null}, we prove new results for the focusing of time-like and null, hypersurface orthogonal congruences.

Although a necessary step towards the proof of the classical singularity theorems, those focusing results are not sufficient to prove singularity formation on a Lorentzian manifold, in general. In order to do so, we shall analyze the energy conservation equation of a given gravity theory. This is especially relevant in the cases of non-metric geodesic curves since, in such cases, one can not rely on the length extremality property of the curves. So, in Section~\ref{Singularity theorems in affine theories of gravity}, we will consider  separately the cases of metric geodesics and accelerated curves. In the latter case, we will first consider two cases of physically motivated energy-momentum tensors: perfect fluids with particular equations of state and scalar fields. In each case, we derive a generalized Raychaudhuri-Komar singularity theorem, applicable to a wide class of affine theories of gravity, without imposing any further symmetries on the manifold. Finally, we will consider the particular case of the Einstein-Cartan theory, allowing us to extend our results to perfect fluids with a cosmological constant and without restrictions on the equation of state.

In this article, we consider Lorentzian metric tensors, which we will use implicitly to lower and raise tensorial indices for the components of a tensor field in a given coordinate basis. Greek indices run from $0$ to $N-1$ and Latin indices run from $1$ to $N-1$, where $N$ denotes the dimension of the manifold.
Moreover, we will assume the metric signature $\left(-+++...\right)$ and follow the conventions for the tensor notation of \cite{Luz}, in particular for the Riemann and torsion tensors.

%%%%%%%%%%%%%%%%%%%%%%%%%%%%%%%%%%%%%%%%%%%%%%%%%%%%%%%%%%%%%%%%%%%%%%%%%%%%%%%%%%%%%%%

\section{Curve focusing in Lorentzian manifolds with torsion}\label{Curve focusing in Lorentzian manifolds with torsion}

In this section, we prove new focus theorems for Lorentzian manifolds endowed with a metric compatible, affine connection. In order to do that we will consider manifolds admitting hypersurface orthogonal congruences of curves.  The results of this section are purely geometric and do not rely on any field equations of a gravity theory.

\subsection{Hypersurface orthogonal congruence\label{sec:Hypersurface-orthogonal-congruence}}
\label{frobeni}

We establish under what conditions can hypersurface orthogonal congruences of curves be locally defined in a Lorentzian manifold with torsion.

Let $\left(M,g,S\right)$ be a Lorentzian manifold of dimension $N\geq2$ with metric $g$ and torsion $S$. Given a point $p\in M$, consider a subset $U_{p}^{*}\subset T_{p}^{*}M$ of dimension $N-m>0$ and $W_{p}\subset T_{p}M$, with dimension $m\ge1$, whose elements $\omega\in W_{p}$ are such that $u\left(\omega\right)=0$, for all $u\in U_{p}^{*}$.

Given an open neighborhood $\mathcal{O}$ of $p\in M$, let $W\subset TM$ be the union of all $W_{q}$, $\forall\,q\in\mathcal{O}$, and $U^{*}$ represent the union of $U_{q}^{*}$. We require $U^{*}$ to be smooth, that is, $U^{*}$ is spanned by $N-m$, $C^{\infty}$ 1-forms. Then, it is well know (Frobenius' theorem) that $W$ admits integral sub-manifolds if and only if, for all $X,Y\in W$ and for all $u\in U^{*}$,
\begin{equation}
u\left(\mathcal{L}_{X}Y\right)=0\,,\label{eq:hyperOrth_condition_Frobenius}
\end{equation}
where $\mathcal{L}_{X}Y$ represents the Lie derivative of the vector field $Y$ along the integral curve of the vector field $X$.
\begin{lem}
\label{Lemma:hypersurface_orth_iff_general}Condition \eqref{eq:hyperOrth_condition_Frobenius} is equivalent to 
\begin{equation}
\nabla_{\left[\alpha\right.}u_{\left.\beta\right]}+S_{\alpha\beta}{}^{\gamma}u_{\gamma}=\sum_{i=1}^{N-m}\left(u_{i}\right)_{\left[\alpha\right.}\left(\chi_{i}\right)_{\left.\beta\right]}\,,\label{eq:hyperOrth_iff_general}
\end{equation}
where $u,\,u_{i}\in U^{*}$ and $\chi_{i}$ are smooth arbitrary 1-forms.
\end{lem}
\begin{proof}
Using the definitions of Lie derivative for vector fields and covariant derivative for a general affine, metric compatible connection we have that 
\[
\left(\mathcal{L}_{X}Y\right)^{\gamma}:=\left[X,Y\right]^{\gamma}=X^{\beta}\nabla_{\beta}Y^{\gamma}-Y^{\beta}\nabla_{\beta}X^{\gamma}-2S_{\alpha\beta}{}^{\gamma}X^{\alpha}Y^{\beta}\,.
\]
Then, from \eqref{eq:hyperOrth_condition_Frobenius}, we find for all $X,Y\in W$ and for all $u\in U^{*}$ 
\[
X^{\alpha}Y^{\beta}\left(\nabla_{\left[\beta\right.}u_{\left.\alpha\right]}-S_{\alpha\beta}{}^{\gamma}u_{\gamma}\right)=0\,.
\]
Since the identity must be valid for all $X$ and $Y$, we get the equality \eqref{eq:hyperOrth_iff_general}.
\end{proof}
Lemma \ref{Lemma:hypersurface_orth_iff_general} states that, locally, $W$ admits integral sub-manifolds if and only if all $u\in U^{*}$ verify Eq.~\eqref{eq:hyperOrth_iff_general}. Two particular cases of interest are when $U^{*}$ is of dimension 1 or 2. In those cases, we have the useful result:
\begin{lem}
\label{Lemma:hypersurface_orth_iff_alt}
Consider a 1-form $u\in U^{*}$ and the associated vector field $u^{\sharp}$, with components $u^{\alpha}=g^{\alpha\beta}u_{\beta}$. If either
\begin{enumerate}
\item $U^{*}$ is of dimension $1$ and $u\left(u^{\sharp}\right)\neq0$,
or
\item $U^{*}$ is of dimension $2$ and $u\left(u^{\sharp}\right)=0$,
\end{enumerate}
then, a congruence of curves, defined in some open set $\mathcal{O}\subset M$ with tangent vector field $u^{\sharp}$, is hypersurface orthogonal if and only if 
\begin{equation}
3u_{\left[\gamma\right.}\nabla_{\beta}u_{\left.\alpha\right]}+u_{\sigma}\left(S_{\beta\alpha}{}^{\sigma}u_{\gamma}+S_{\alpha\gamma}{}^{\sigma}u_{\beta}+S_{\gamma\beta}{}^{\sigma}u_{\alpha}\right)=0\,.\label{eq:hyperOrth_iff_alt}
\end{equation}
\end{lem}
The proof of Lemma \ref{Lemma:hypersurface_orth_iff_alt} follows directly from Eq.~\eqref{eq:hyperOrth_iff_general} in the cases when $N-m=1$ and $N-m=2$, for a particular choice of the 1-forms $\chi_{i}$. Eq.~\eqref{eq:hyperOrth_iff_alt} has the advantage of being independent of the arbitrary 1-forms in Eq.~\eqref{eq:hyperOrth_iff_general}.

%%%%%%%%%%%%%%%%%%%%%%%%%%%%%%%%%%%%%%%%%%%%%%%%%%%%%%%%%%%%%%%%%%%%%%%%%%%%

\subsection{Focusing of time-like congruences\label{sec:Caustics_timelike}}

We will now study necessary conditions for the formation of a focal point of a hypersurface orthogonal congruence. So far, our results are rather general, being valid for any type of curve, namely space-like, time-like or null, even for curves non-affinely parameterized. To proceed, however, we will need to specify the type of curve that is being considered. Let us then first study the case of time-like curves.

Consider a congruence of time-like curves with tangent vector field $v$ with acceleration
\begin{equation}
a^{\alpha}:=\nabla_{v}v^{\alpha}\,. \label{timelike_acceleration}
\end{equation}
Without loss of generality, we will assume $v$ to be affinely parameterized, i.e.
\begin{equation}
\begin{aligned}v^{\alpha}v_{\alpha} & =-1\,,\\
v_{\alpha}a^{\alpha} & =0\,.
\end{aligned}
\end{equation}
In what follows, it will be useful to adapt to our case the $1+3$ formalism using the notation of Refs.~\cite{Ehlers,Ellis}. 
In particular, we will decompose the torsion tensor field as a sum of its orthogonal and tangent components to the congruence. To do so, we will make use of the projector $h$ onto the $N-1$ dimensional subspace orthogonal to the time-like congruence, defined as 
\begin{equation}
h_{\alpha\beta}=g_{\alpha\beta}+v_{\alpha}v_{\beta}\,,\label{eq:timelike_projector_def}
\end{equation}
with the following properties 
\begin{equation}
\begin{aligned}h_{\alpha\beta}v^{\alpha} & =0\,,\\
h_{\alpha\beta}h^{\alpha\gamma} & =h_{\beta}^{~\gamma}\,,\\
h_{\alpha\beta}h^{\alpha\beta} & =N-1.
\end{aligned}
\label{eq:timelike_projector_properties}
\end{equation}
Using Eq.~\eqref{eq:timelike_projector_def}, a straightforward calculation yields the following expression for the torsion tensor field: 
\begin{equation}
S_{\alpha\beta\gamma}=\varepsilon_{\alpha\beta}\,^{\mu}\bar{S}_{\mu\gamma}+W_{[\alpha|\gamma}v_{|\beta]}+S_{\alpha\beta}v_{\gamma}+v_{[\alpha}A_{\beta]} v_{\gamma}\,,\label{eq:timelike_torsion_decomposition}
\end{equation}
where 
\begin{equation}
\begin{aligned}\bar{S}_{\alpha\beta} &
=\frac{1}{2}\varepsilon_{\alpha\mu\nu}h_{~\beta}^{\sigma}S^{\mu\nu}{}_{\sigma}\,,\\
W_{\alpha\beta} &
=2v^{\mu}h_{~\alpha}^{\nu}h_{~\beta}^{\sigma}S_{\mu\nu\sigma}\,,\\
S_{\alpha\beta} &
=-h_{~\alpha}^{\mu}h_{~\beta}^{\nu}v^{\sigma}S_{\mu\nu\sigma}\,,\\
A_{\alpha} &
=2v^{\mu}h_{~\alpha}^{\nu}v^{\sigma}S_{\mu\nu\sigma}\,,
\end{aligned}
\label{eq:timelike_torsion_decomposition_components_def}
\end{equation}
and $\varepsilon_{\alpha\beta\gamma}:=\varepsilon_{\alpha\beta\gamma\delta}v^{\delta}$ represents the projected covariant Levi-Civita tensor. It is worth remarking that the tensors defined in Eq.~\eqref{eq:timelike_torsion_decomposition_components_def} are orthogonal to the tangent vector $v$.

Given the definitions of the expansion scalar, shear tensor and vorticity tensor of the congruence associated with the vector field $v$ as \cite{Luz}
\begin{align}
\theta &
:=\nabla_{\mu}v^{\mu}+2S_{\sigma\mu}{}^{\mu}v^{\sigma}\,,\label{eq:timelike_expansion_def}\\
\sigma_{\alpha\beta} &
:=h^{\mu}{}_{\left(\alpha\right.}h_{\left.\beta\right)}\,^{\nu}\left(\nabla_{\mu}v_{\nu}+2S_{\sigma\mu\nu}v^{\sigma}\right)\,,\\
\omega_{\alpha\beta} &
:=h^{\mu}{}_{\left[\alpha\right.}h_{\left.\beta\right]}\,^{\nu}\left(\nabla_{\mu}v_{\nu}+2S_{\sigma\mu\nu}v^{\sigma}\right)\,,\label{eq:timelike_vorticity_def}
\end{align}
we can show:
\begin{lem}
\label{Lemma:lemma_2.2_1}
Given an affinely parametrized congruence of time-like curves in a Lorentzian manifold with torsion $(M,g,S)$, a necessary and sufficient condition for the congruence to be hypersurface orthogonal is that the vorticity tensor is given by 
\begin{equation}
\omega_{\alpha\beta}=W_{\left[\alpha\beta\right]}+S_{\alpha\beta},\label{eq:timelike_hyperOrth_vorticity}
\end{equation}
where $W_{\alpha\beta}$ and $S_{\alpha\beta}$ are defined in \eqref{eq:timelike_torsion_decomposition_components_def}.
\end{lem}
\begin{proof}
First, by using Eqs.~\eqref{eq:timelike_torsion_decomposition} and \eqref{eq:timelike_torsion_decomposition_components_def}, we find that Eq.~\eqref{eq:timelike_hyperOrth_vorticity} is equivalent to
\begin{equation}
\omega_{\alpha\beta}-2v^{\mu}h_{~[\alpha}^{\nu}h_{\beta]}^{~\sigma}S_{\mu\nu\sigma}+h_{~\alpha}^{\mu}h_{~\beta}^{\nu}v^{\sigma}S_{\mu\nu\sigma}=0\,.\label{eq:timelike_hyperOrth_imply_proof}
\end{equation}
To show that Eq.~\eqref{eq:timelike_hyperOrth_imply_proof} is a necessary and sufficient condition for the congruence to admit orthogonal hypersurfaces, it suffices to prove that this equation is equivalent to Eq.~\eqref{eq:hyperOrth_iff_general} for dimensions $N>1$ and $N-m=1$. In fact, contracting Eq.~\eqref{eq:hyperOrth_iff_general} with $h_{~\mu}^{\alpha}h_{~\nu}^{\beta}$ and using Eq.~\eqref{eq:timelike_vorticity_def}, we recover Eq.~\eqref{eq:timelike_hyperOrth_imply_proof}. On the other hand, substituting Eq.~\eqref{eq:timelike_projector_def} in Eq.~\eqref{eq:timelike_vorticity_def} and using Eq.~\eqref{eq:timelike_hyperOrth_imply_proof} we find \eqref{eq:hyperOrth_iff_general}.
\end{proof}
Gathering the previous results, we are in position to prove a generalized focus theorem for time-like curves in manifolds with torsion: 
\begin{thm}
\label{Theorem:timelike_focus_theorem}
Let $\left(M,g,S\right)$ be a $N$-dimensional Lorentzian manifold endowed with a torsion tensor field $S$ such that $W_{\left(\alpha\beta\right)}=0$. Given a congruence of hypersurface orthogonal time-like curves in $\left(M,g,S\right)$, let $c$ represent an affinely parametrized fiducial curve of the congruence, with tangent vector field $v$. A focal point of the congruence will form, for a finite value of the affine parameter, if the following three conditions are satisfied:\leavevmode 
\begin{enumerate}
\item At a point along $c$, the expansion of the congruence is negative;
\item
$R_{\alpha\beta}v^{\alpha}v^{\beta}\geq\left(W_{\left[\alpha\beta\right]}+S_{\alpha\beta}\right)S^{\alpha\beta}$;
\item $\nabla_{\alpha}a^{\alpha}+A_{\alpha}a^{\alpha}\leq0$,
\end{enumerate}
where $R_{\alpha\beta}$ is the Ricci tensor.
\end{thm}
\begin{proof}
For a hypersurface orthogonal congruence, the Raychaudhuri equation \cite{Luz} can be written as 
\begin{equation}
\begin{aligned}\nabla_{v}\theta=\frac{d\theta}{d\tau}:=\dot{\theta}
&
=-R_{\alpha\beta}v^{\alpha}v^{\beta}-\left(\frac{1}{N-1}\theta^{2}+\sigma_{\alpha\beta}\sigma^{\alpha\beta}\right)+\left(W_{\left[\alpha\beta\right]}+S_{\alpha\beta}\right)S^{\alpha\beta}+\\
&
+W^{\left(\alpha\beta\right)}\left(\frac{h_{\alpha\beta}}{N-1}\theta+\sigma_{\alpha\beta}\right)+\nabla_{\alpha}a^{\alpha}+\dot{W}_{\alpha}{}^{\alpha}+A_{\alpha}a^{\alpha}\,,
\end{aligned}
\label{eq:timelike_raychaudhuri_general_proof}
\end{equation}
where $\tau$ is an affine parameter. Assuming $W_{\left(\alpha\beta\right)}=0$, Eq.~\eqref{eq:timelike_raychaudhuri_general_proof} then reduces to 
\begin{equation}
\begin{aligned}\dot{\theta} &
=-R_{\alpha\beta}v^{\alpha}v^{\beta}-\left(\frac{1}{N-1}\theta^{2}+\sigma_{\alpha\beta}\sigma^{\alpha\beta}\right)+\left(W_{\left[\alpha\beta\right]}+S_{\alpha\beta}\right)S^{\alpha\beta}+\nabla_{\alpha}a^{\alpha}+A_{\alpha}a^{\alpha}\,.\end{aligned}
\label{eq:timelike_raychaudhuri_general_proof_2}
\end{equation}
Since $\sigma_{\alpha\beta}$ is orthogonal to $v$, we have $\sigma_{\alpha\beta}\sigma^{\alpha\beta}\geq0$, therefore, under the conditions of the theorem, 
\begin{equation}
\frac{d\theta}{d\tau}\leq-\frac{\theta^{2}}{N-1}
\end{equation}
which upon integration implies 
\begin{equation}
\frac{1}{\theta\left(\tau\right)}\geq\frac{\tau}{N-1}+\frac{1}{\theta_{0}}\,,
\end{equation}
where $\theta_{0}$ is an initial value for $\theta$.

Hence, assuming the existence of a trapped surface where $\theta_{0}<0$, for a value $\tau\leq\left(N-1\right)/|\theta_{0}|$, then $\theta\to-\infty$ i.e. a focal point of the congruence will form. 
\end{proof}
\begin{rem}
Two important particular cases that are widely used in the physics literature and verify the conditions in Theorem \ref{Theorem:timelike_focus_theorem} are: Torsion tensors of the form
$S_{\alpha\beta}{}^{\sigma}=S_{\alpha\beta}v^{\sigma}$ (see e.g. \cite{Trautman,Hehl2,Hadi,Boehmer,Mena} and references therein); Torsion tensors completely anti-symmetric
$S_{\alpha\beta\sigma}=S_{\left[\alpha\beta\sigma\right]}$, where $A^{\alpha}=0$ and $W_{\alpha\beta}=W_{\left[\alpha\beta\right]}=-2S_{\alpha\beta}$ (see e.g. \cite{Visser,Hammad,Fabbri}).
\end{rem}
To close this section, we make a few comments about the constraints on the torsion tensor in the above theorem. In the presence of a general torsion field, the Raychaudhuri equation has a very different functional form when compared to the case of a manifold without torsion. This added complexity makes it very difficult to infer the sign of the term with $W^{\left(\alpha\beta\right)}$ in Eq.~\eqref{eq:timelike_raychaudhuri_general_proof} without knowing, \emph{ab initio}, the quantities $\theta$ and $\sigma_{\alpha\beta}$. Therefore, in the case when $W^{\left(\alpha\beta\right)}\neq0$ it is not possible to make a general statement regarding the formation of a focal point taking into account only the Raychaudhuri equation.

%%%%%%%%%%%%%%%%%%%%%%%%%%%%%%%%%%%%%%%%%%%%%%%%%%%%%%%%%%%%%%%%%%%%

\subsection{Focusing of null congruences\label{sec:Caustics_null}}

We now consider the case of a congruence of null curves with tangent vector field $k$. Moreover, in this subsection, we will consider that the manifold is of dimension $N\geq3$. Assuming $k$ to be affinely parameterized, we have 
\begin{equation}
\begin{aligned}k^{\alpha}k_{\alpha} & =0\,,\\
\nabla_{k}k^{\alpha} & =w^{\alpha}\,,\\
k_{\alpha}w^{\alpha} & =0\,.
\end{aligned}
\label{eq:Caustics_null_tangent_vector_properties}
\end{equation}
We will follow a similar procedure to Section \ref{sec:Caustics_timelike}. However, instead of using an adaptation of the $1+3$ formalism, we shall use a $2+2$ type decomposition adapted to $N$-dimensional manifolds. To do so, we will have to find the proper projector to the subspace orthogonal to the tangent vector $k$. Since $k$ is null, it is orthogonal to itself and, as such, the object in Eq.~\eqref{eq:timelike_projector_def} does not meet the requirements for such projector. Let us then consider an auxiliary vector field $\xi$, with the following properties 
\begin{equation}
\begin{aligned}\xi^{\alpha}\xi_{\alpha} & =0\,,\\
\xi^{\alpha}k_{\alpha} & =-1.
\end{aligned}
\label{eq:Caustics_null_auxiliar_vector_properties}
\end{equation}
Then, we can define the following tensor 
\begin{equation}
\tilde{h}_{\alpha\beta}:=g_{\alpha\beta}+k_{\alpha}\xi_{\beta}+\xi_{\alpha}k_{\beta}\,,\label{eq:null_projector_def}
\end{equation}
such that
\begin{equation}
\begin{aligned}\tilde{h}_{\alpha\beta}k^{\alpha} & =0\,,\\
\tilde{h}_{\alpha\beta}\xi^{\alpha} & =0\,,\\
\tilde{h}_{\alpha\beta}\tilde{h}^{\alpha\gamma} &
=\tilde{h}_{\beta}^{\gamma}\,,\\
\tilde{h}_{\alpha\beta}\tilde{h}^{\alpha\beta} & =N-2\,,
\end{aligned}
\label{eq:null_projector_properties}
\end{equation}
where $\tilde{h}$ represents the projector onto the $N-2$ subspace orthogonal to both $k$ and $\xi$.

The vorticity of the congruence associated with $k$ on the $N-2$ subspace, orthogonal to both $k$ and $\xi$ at each point, is given by \cite{Luz}
\begin{align}
\theta &
:=\nabla_{\alpha}k^{\alpha}-2S_{\alpha\gamma}\,^{\alpha}k^{\gamma}-2\xi^{\sigma}S_{\sigma\gamma\rho}k^{\gamma}k^{\rho}+\xi^{\gamma}w_{\gamma}\,,\\
\sigma_{\alpha\beta} &
:=\tilde{h}^{\mu}{}_{\left(\alpha\right.}\tilde{h}_{\left.\beta\right)}\,^{\nu}\left(\nabla_{\mu}k_{\nu}+2S_{\sigma\mu\nu}k^{\sigma}\right)\,,\\
\omega_{\alpha\beta} &
:=\tilde{h}^{\mu}{}_{\left[\alpha\right.}\tilde{h}_{\left.\beta\right]}\,^{\nu}\left(\nabla_{\mu}k_{\nu}+2S_{\sigma\mu\nu}k^{\sigma}\right)\,.\label{eq:null_vorticity_def}
\end{align}
Using the results in Section
\ref{sec:Hypersurface-orthogonal-congruence},
we can now obtain necessary and sufficient conditions for the
congruence
associated with $k$ to be hypersurface orthogonal:
\begin{lem}
\label{Lemma:lemma_2.3_1}Given an affinely parametrized
congruence
of null curves in a Lorentzian manifold with torsion $(M,g,S)$,
of
dimension $N\geq3$, a necessary and sufficient condition for the congruence to be hypersurface orthogonal is that the vorticity
tensor
is characterized by
\begin{equation}
\omega_{\alpha\beta}=2\tilde{h}^{\mu}{}_{\left[\alpha\right.}\tilde{h}_{\left.\beta\right]}{}^{\nu}S_{\sigma\mu\nu}k^{\sigma}-\tilde{h}_{~\alpha}^{\mu}\tilde{h}_{~\beta}^{\nu}S_{\mu\nu\sigma}k^{\sigma}\,.\label{eq:null_hyperOrth_vorticity}
\end{equation}
\end{lem}
\begin{proof}
As before, we give a proof sketch without including all the algebraic manipulations. In the considered case, Eq.~\eqref{eq:hyperOrth_iff_general} reads
\begin{equation}
\nabla_{\left[\alpha\right.}k_{\left.\beta\right]}+S_{\alpha\beta}{}^{\gamma}k_{\gamma}=k_{\left[\alpha\right.}\left(\chi_{1}\right)_{\left.\beta\right]}+\xi_{\left[\alpha\right.}\left(\chi_{2}\right)_{\left.\beta\right]}\,,\label{eq:null_prop4.1_proof_1}
\end{equation}
where $\chi_{1}$ and $\chi_{2}$ are arbitrary 1-forms. Contracting Eq.~\eqref{eq:null_prop4.1_proof_1} with $\tilde{h}^{~\alpha}_{\mu}\tilde{h}^{~\beta}_{\nu}$ and using the properties \eqref{eq:null_projector_properties} and Eq.~\eqref{eq:null_vorticity_def}, it is straightforward to find Eq.~\eqref{eq:null_hyperOrth_vorticity}. On the other hand, starting from Eq.~\eqref{eq:null_hyperOrth_vorticity}, using Eqs.~\eqref{eq:null_projector_def} and \eqref{eq:null_vorticity_def} we recover \eqref{eq:null_prop4.1_proof_1}. This then shows that Eq.~\eqref{eq:null_hyperOrth_vorticity} is equivalent to Eq.~\eqref{eq:null_prop4.1_proof_1}. 
\end{proof}
We note that Eq.~\eqref{eq:null_hyperOrth_vorticity} and some subcases, derived in other contexts, can already be found in the literature (cf. Refs.~\cite{Liberati,Hammad}).

Let us now write the Raychaudhuri equation for the congruence associated with $k$, valid on the $N-2$ subspace orthogonal to both $k$ and $\xi$. Assuming, without lack of generality, that the congruence of curves associated with $k$ are affinely parameterized by a parameter $\lambda$, we find 
\begin{equation}
\begin{aligned}\nabla_{k}\theta:=\frac{d\theta}{d\lambda} &
=-R_{\alpha\beta}k^{\alpha}k^{\beta}-\left(\frac{1}{N-2}\theta^{2}+\sigma_{\alpha\beta}\sigma^{\alpha\beta}+\omega_{\alpha\beta}\omega^{\beta\alpha}\right)+\\
&
+2S_{\rho\alpha}{}^{\beta}k^{\rho}\left(\frac{\tilde{h}_{\beta}{}^{\alpha}}{N-2}\theta+\sigma_{\beta}{}^{\alpha}+\omega_{\beta}{}^{\alpha}\right)+2k^{\gamma}\nabla_{\gamma}\left(S_{\rho}k^{\rho}\right)+\nabla_{\alpha}w^{\alpha}+\\
&
+k^{\gamma}\nabla_{\gamma}\left(\xi^{\alpha}w_{\alpha}\right)+\xi^{\alpha}\xi^{\beta}w_{\alpha}w_{\beta}+2\xi^{\beta}w^{\alpha}\nabla_{\alpha}k_{\beta}+2S_{\beta\gamma\rho}k^{\beta}\xi^{\rho}w^{\gamma}\\
&
+2k^{\gamma}\nabla_{\gamma}\left(S_{\rho\alpha\sigma}\,k^{\rho}k^{\sigma}\xi^{\alpha}\right)+2S_{\alpha\sigma\beta}k^{\alpha}k^{\beta}\xi^{\gamma}\nabla_{\gamma}k^{\sigma}\,.
\end{aligned}
\label{eq:null_Raychaudhuri_general}
\end{equation}
For a completely general torsion tensor field and $k$ vector field, there seems to be no relevant way to constraint Eq.~\eqref{eq:null_Raychaudhuri_general} in order to infer the behavior of the expansion of the congruence associated with $k$.

On the other hand, one could argue, on physical grounds, that Eq.~\eqref{eq:null_Raychaudhuri_general} is only useful if $k$ describes the paths of physical particles, hence, $k$ must be tangent to metric geodesics (cf. refs.~\cite{Luz,obukhov}), i.e. $k^{\alpha}\nabla_{\alpha}^{\left(g\right)}k^{\beta}=0$, where $\nabla_{\alpha}^{\left(g\right)}$ represents the Levi-Civita connection. From the definition of $w^{\alpha}$, such choice requires $w^{\alpha}=2S^{\alpha}{}_{\gamma\sigma}k^{\gamma}k^{\sigma}$. In the literature (see e.g. \cite{Liberati,Hammad}), there has been also interest in the case when $k^{\alpha}\nabla_{\alpha}k^{\beta}=0$, which implies $w^{\alpha}=0$. Unfortunately, in both cases, for our purposes, Eq.~\eqref{eq:null_Raychaudhuri_general} still contains terms that can only be evaluated if we have full knowledge of the evolution of the metric and torsion tensors\footnote{Consider, for instance, the last terms in Eq.~\eqref{eq:null_Raychaudhuri_general} such as $\xi^{\gamma}\nabla_{\gamma}k^{\sigma}$ and $k^{\gamma}\nabla_{\gamma}\xi^{\sigma}$.}. Therefore, for a completely general torsion tensor field, there seems to be no mathematically interesting way to construct a useful focus theorem based only on the Raychaudhuri equation. Notice that, although the auxiliary vector $\xi$ is not completely determined by Eq.~\eqref{eq:Caustics_null_auxiliar_vector_properties}, we do not have enough freedom to remove all the problematic terms (whose signs we are not able to constraint \emph{ab initio}). Moreover, requiring hypersurface orthogonality of the congruence associated with $k$ (which we haven't done yet in this section) does not solve this problem.

Based on the discussion in previous paragraph, in the remaining of this section, we will consider the case when the torsion tensor field is completely anti-symmetric and the ``acceleration'' $w$
is zero: 
\begin{equation}
S_{\alpha\beta\sigma}=S_{\left[\alpha\beta\sigma\right]}\quad\text{and}\quad
w=0.
\end{equation}
These assumptions allow us to describe both the cases $k^{\alpha}\nabla_{\alpha}k^{\beta}=0$ and $k^{\alpha}\nabla_{\alpha}^{\left(g\right)}k^{\beta}=0$ at once. We have then the following result: 
\begin{thm}
\label{Theorem:null_focus_theorem} Let $\left(M,g,S\right)$ be a Lorentzian manifold of dimension $N\geq3$ endowed with a totally anti-symmetric torsion tensor field $S$. Given a congruence of hypersurface orthogonal null curves in $\left(M,g,S\right)$, let $c$ represent the fiducial curve of the congruence, with tangent vector field $k$ such that $\nabla_{k}k=0$. A focal point of the congruence will form, for a finite value of the parameter of $c$, if the following two conditions are satisfied:\leavevmode 
\begin{enumerate}
\item At a point along $c$, the expansion of the congruence becomes negative;
\item $R_{\alpha\beta}k^{\alpha}k^{\beta}+2S_{\alpha\beta\mu}S^{\alpha\beta\nu}k^{\mu}k_{\nu}\geq0$.
\end{enumerate}
\end{thm}
\begin{proof}
As in Section \ref{sec:Caustics_timelike}, we will assume without loss of generality that the fiducial curve is affinely parameterized by a parameter $\lambda$. Given the premises of the theorem, the generalized Raychaudhuri equation \eqref{eq:null_Raychaudhuri_general} can be written as 
\begin{equation}
\frac{d\theta}{d\lambda}=-R_{\alpha\beta}k^{\alpha}k^{\beta}-2S_{\alpha\beta\mu}S^{\alpha\beta\nu}k^{\mu}k_{\nu}-\left(\frac{1}{N-2}\theta^{2}+\sigma_{\alpha\beta}\sigma^{\alpha\beta}\right)\,.
\end{equation}
Repeating the reasoning of the proof of Theorem \ref{Theorem:timelike_focus_theorem} we find 
\begin{equation}
\frac{1}{\theta\left(\lambda\right)}\geq\frac{\lambda}{N-2}+\frac{1}{\theta_{0}}\,,
\end{equation}
where $\theta_{0}$ is the initial value of $\theta$. Hence, assuming the existence of a trapped surface where $\theta_{0}<0$, for a value $\lambda\leq\left(N-2\right)/|\theta_{0}|$, then $\theta\to-\infty$, i.e. a focal point of the congruence will form for a finite value of $\lambda$.
\end{proof}
%%%%%%%%%%%%%%%%%%%%%%%%%%%%%%%%%%%%%%%%%%%%%%%%%%%%%%%%%%%%%%%%%

\section{Singularity theorems in affine theories of gravity}\label{Singularity theorems in affine theories of gravity}

The previous results concerned the formation of focal points of congruences of curves on Lorentzian manifolds. Nonetheless, the existence of such points does not necessarily imply the formation of singularities. In order to prove the existence of a singularity on the manifold, the formation of a focal point needs to be related to some kind of incompleteness of the manifold structure. To approach this problem there has been, at least, two lines of works (cf. e.g. \cite{Senovilla_1998}) in the sense that the occurrence of a focal point has been related either to the incompleteness of causal curves or, given a physical theory of gravity, to the lack of regularity of some physical scalar invariant such as the energy density. In the former approach, metric geodesics are natural curves to consider, from the geometric point of view, as they are curves of extremal length according to the metric tensor $g$. On the other hand, if one considers non-geodesic curves, it can be more feasible to follow the latter approach.

%%%%%%%%%%%%%%%%%%%%%%%%%%%%%%%%%%%%%%%%%%%%%%%%%%%%%%%%%%%%%%%%%

\subsection{Geodesic curves}

In Section \ref{sec:Caustics_null}, we considered the case when the torsion tensor field is totally anti-symmetric and $\nabla_{k}k=0$. In such case, the integral curves of the null vector $k$ are metric geodesics, therefore, the results of the singularity theorems of Hawking and Penrose can be applied, assuming instead the conditions of  Theorem~\ref{Theorem:null_focus_theorem}. In fact, the argument also holds for time-like geodesics\footnote{This was first noticed in \cite{Hehl1}, in the context of the Einstein-Cartan theory.}. We summarize these observations as:
\begin{prop}
\label{Lemma:HP_like_theorem}
Given a congruence of time-like (resp. null) metric geodesics in $\left(M,g,S\right)$, if the conditions of Theorem~\ref{Theorem:timelike_focus_theorem} (resp. Theorem \ref{Theorem:null_focus_theorem}) are verified, then the geodesics of the congruence are incomplete.
\end{prop}
The results of Proposition~\ref{Lemma:HP_like_theorem} imply, in particular, that the presence of torsion modifies the so-called \emph{strong energy condition} of General Relativity, by including the explicit contribution of torsion in conditions 2 of Theorems ~\ref{Theorem:timelike_focus_theorem} and ~\ref{Theorem:null_focus_theorem}. This observation was  already behind the work of Trautman \cite{Trautman} (see also \cite{Stewart, Kop,Hehl2}) where it was found, in a particular model of the Einstein-Cartan theory, that the presence of torsion could prevent the formation of singularities due to the violation of the modified strong energy condition.

%%%%%%%%%%%%%%%%%%%%%%%%%%%%%%%%%%%%%%%%%%%%%%%%%%%%%%%%%%%%%%%%%%%%%%%%%%%%%%%%%%%%%%

\subsection{Non-geodesic curves\label{sec:Singularity_theorems}}

An affine theory of gravity is characterized by field equations that relate the geometry of the manifold with the matter fields that permeate it. In general, at least one set of these field equations can be written in the form 
\begin{equation}
f\left(g,S,T,\text{other fields}\right)=0\,,\label{eq:Field_equations_general}
\end{equation}
where $f$ is some functional, $g$ represents the metric tensor, $S$ the torsion tensor and $T$ is a rank-2 tensor constructed from the matter fields which we call the \emph{stress-energy tensor} and whose divergence verifies 
\begin{equation}
\nabla_{\beta}T^{\alpha\beta}=\Psi^{\alpha}\,,\label{eq:EM_tensor_divergence_definition}
\end{equation}
for some vector field $\Psi$. 

Before we write the main result of this subsection, let us address the \emph{other fields} part in Eq.~\eqref{eq:Field_equations_general}. Here, we are thinking of theories like, for instance, scalar-tensor theories with a null torsion tensor field, where an extra scalar field non-minimally coupled to the geometry is considered and appears directly in the field equations. In those theories, it is still possible to write a metric stress-energy tensor for the matter fields so that, in the so-called Jordan frame, Eq.~\eqref{eq:EM_tensor_divergence_definition} is verified for an identically null $\Psi$ (cf. e.g. Ref.~\cite{Santiago}).

Considering a gravity theory on a Lorentzian manifold, which we now call {\em space-time}, we use the following definition:\footnote{In his review \cite{Senovilla_1998}, Senovilla calls {\em matter singularity} to a singularity in the energy density term of the stress-energy tensor.}
\begin{defn}
\label{Definition:Physical_singularity} A space-time $(M,g,S)$ is said to be physically singular if there exists at least one frame (i.e. an observer) for which the stress-energy tensor $T$ is not well defined in $(M,g,S)$, for any coordinate system. 
\end{defn}
We remark that, classically, a stress-energy tensor consists of a rank-2 tensor with real entries. So, if independently of the coordinate system, some components of $T$ either diverge or take complex values, this would correspond to a non-admissible physical object, in the sense that an observer is no longer able to make predictions using that physical theory.

Now, large classes of physical particles follow accelerated time-like curves and the effects of the acceleration might not be negligible in the extreme gravitational field regime. From a physical point of view, a congruence of curves represents the paths of matter particles, as such, a focal point of the congruence represents the focusing of a matter field. It is then natural to connect the formation of a focal point with the stress-energy tensor that characterizes the matter fluid that permeates the space-time. A prototype of this idea was first used by Raychaudhuri \cite{Raychaudhuri_1955} and, independently, by Komar \cite{Komar_1956}, and later formalized by Senovilla \cite{Senovilla_1998} for the case of congruences of time-like geodesics in torsionless space-times, in the context of GR. Here, we extend the results of Ref.~\cite{Senovilla_1998} to time-like congruences in space-times endowed with a torsion tensor field, in the context of affine theories of gravity.

For a congruence of time-like curves with tangent vector $v$, we can decompose $T$ as 
\begin{equation}
\begin{aligned}T_{\alpha\beta} &
=\rho\,v_{\alpha}v_{\beta}+p\,h_{\alpha\beta}+q_{1\alpha}v_{\beta}+v_{\alpha}q_{2\beta}+\pi_{\alpha\beta}+m_{\alpha\beta}\end{aligned}
\,,\label{eq:Energy_momentum_tensor_decomposition_general-2}
\end{equation}
with 
\begin{equation}
\begin{aligned}\rho & =v^{\mu}v^{\nu}T_{\mu\nu}\,, &
\hspace{1.5cm} & & q_{1\alpha} &
=-h_{\alpha}^{~\mu}v^{\nu}T_{\mu\nu}\,,\\
p & =\frac{1}{N-1}h^{\mu\nu}T_{\mu\nu}\,, & & & q_{2\alpha} &
=-v^{\mu}h_{\alpha}^{~\nu}T_{\mu\nu}\,,\\
\pi_{\alpha\beta} &
=\left(h^{\mu}_{\,\left(\alpha\right.}h^{\nu}_{\left.\beta\right)}{}-\frac{h_{\alpha\beta}}{N-1}h^{\mu\nu}\right)T_{\mu\nu}\,,
& & & m_{\alpha\beta} &
=h_{\,\left[\alpha\right.}^{\mu}h^{\nu}_{\left.\beta\right]}{}T_{\mu\nu}\,,
\end{aligned}
\label{eq:Energy_momentum_tensor_decomposition_quantities-1}
\end{equation}
where $N\geq2$ represents the dimension of the space-time. 

Before proceeding, notice that the components in Eq.~\eqref{eq:Energy_momentum_tensor_decomposition_quantities-1} are defined covariantly such that, if it is found that some component is not defined at some point of the manifold, this indeed means that an observer co-moving with the fiducial curve of the congruence is unable to define a stress-energy tensor (in contrast with just an ill choice of coordinates).

Now, given Definition \ref{Definition:Physical_singularity}, we obtain the following result:
\begin{thm}
\label{Theorem:singularities} Consider a space-time $(M,g,S)$ in the context of any affine theory of gravity satisfying the conditions of Theorem \ref{Theorem:timelike_focus_theorem}, the conservation of energy $v^{\alpha}\nabla_{\beta}T_{~\alpha}^{\beta}=0$ and containing either
\begin{enumerate}
\item A fluid with $T_{\alpha\beta}=\rho v_{\alpha}v_{\beta}+p\,h_{\alpha\beta}$ and $p=k\rho^{\gamma}$, where $k\in\mathbb{R}_{0}^+$, $\gamma\in\mathbb{R}\setminus\{0\}$ and, at some instant, $\tau=\tau_0$, $\theta(\tau_0)<0$ and $\rho(\tau_0)>0$;
\item A scalar field with $T_{\alpha\beta}=\left(\nabla_{\alpha}\psi\right)\left(\nabla_{\beta}\psi\right)-\frac{1}{2}g_{\alpha\beta}\left(\nabla_{\mu}\psi\right)\left(\nabla^{\mu}\psi\right)-g_{\alpha\beta}V\left(\psi\right),$ admitting level surfaces, such that its gradient $\nabla_{\alpha}\psi$ is time-like,
\end{enumerate}
then, the space-time will become physically singular at, or before, the formation of a focal point. 
\end{thm}
\begin{cor}
\label{Proposition:EC_perf_fluid_non_existence_caustic}
In case 1, if the fluid has $\gamma>1$, then the space-time becomes physically singular before the formation of a focal point. 
\end{cor}
\begin{cor}
Theorem \ref{Theorem:singularities} applies, in particular, to any minimally coupled affine theory of gravity without torsion, to scalar-tensor theories written in the Jordan frame in manifolds without torsion and to the Einstein-Cartan theory for a torsion given by $S_{\alpha\beta\gamma}=S_{\alpha\beta}v_{\gamma}+v_{[\alpha}A_{\beta]}v_{\gamma}$.
\end{cor}
\begin{proof}
Using Eqs.~\eqref{eq:timelike_projector_def} and \eqref{eq:Energy_momentum_tensor_decomposition_general-2}, then $v^{\alpha}\nabla_{\beta}T_{~\alpha}^{\beta}=0$ can be written as
\begin{align}
\nabla_{v}\rho+\theta\left(\rho+p\right)+\pi^{\alpha\beta}\sigma_{\alpha\beta}-m^{\alpha\beta}\omega_{\alpha\beta}+q_{1}^{\alpha}a_{\alpha}+\nabla_{\alpha}q_{2}^{\alpha}
& =0\,,\label{eq:Conservation_timelike}
\end{align}
where $\theta$, $\sigma_{\alpha\beta}$ and $\omega_{\alpha\beta}$ are given by Eqs.~\eqref{eq:timelike_expansion_def} - \eqref{eq:timelike_vorticity_def}.
\\\\
(i)~Considering the particular case of a fluid characterized by a stress-energy tensor whose non-zero components in the decomposition \eqref{eq:Energy_momentum_tensor_decomposition_general-2} are $\rho$ and $p$, then Eq.~\eqref{eq:Conservation_timelike} reduces to 
\begin{align}
\nabla_{v}\rho+\theta\left(\rho+p\right) &
=0.\label{eq:conservation_timelike_perf_fluid}
\end{align}
Assuming $p=k\rho^{\gamma}$, with $k\in\mathbb{R}_{0}^+$ and $\gamma\in\mathbb{R}\setminus\{0\}$, we may integrate Eq.~\eqref{eq:conservation_timelike_perf_fluid} along the fiducial curve of the congruence to find 
\begin{equation}
\begin{cases}
\left|\rho \right|=\left|\rho_{0} \right|e^{-\left(1+k\right)\int_{\tau_{0}}^{\tau}\theta\left(t\right)dt}\,,
& \text{for \ensuremath{\gamma=1}}\\
\left|\rho{}^{1-\gamma} +k \right|=\left(\rho_{0}^{1-\gamma}+k \right)\,e^{-\left(1-\gamma\right)\int_{\tau_{0}}^{\tau}\theta\left(t\right)dt}\,,
& \text{for \ensuremath{\gamma\neq1}},
\end{cases}\label{eq:timelike_perf_fluid_rho_expressions}
\end{equation}
where $\rho\equiv\rho\left(c\left(\tau\right)\right)$ and
$\rho_{0}\equiv\rho\left(c\left(\tau_{0}\right)\right)$.

Assuming that the conditions of the focus Theorem \ref{Theorem:timelike_focus_theorem} are verified, depending on the values of $\gamma$, we get different types of behavior. If $\gamma\le1$, taking the limit of Eq.~\eqref{eq:timelike_perf_fluid_rho_expressions}, we find 
\begin{equation}
\lim_{\tau\to\tau_{f}^{-}}\left|\rho\right|=+\infty\text{ , for
\ensuremath{\gamma\le1}},\label{eq:Perfect_fluid_limit_rho_gamma_leq_1-1}
\end{equation}
i.e., the space-time becomes physically singular at the formation of the focal point. On the other hand, if $\gamma>1$ and assuming $\rho_{0}>0$, from Eq.~\eqref{eq:timelike_perf_fluid_rho_expressions} we find that $\rho^{1-\gamma}$ is a strictly decreasing function and 
\begin{equation}
\lim_{\tau\to\tau_{a}^{-}}\frac{1}{\rho^{\gamma-1}}=0,
\end{equation}
for some $\tau_{a}<\tau_{f}$. So, before the formation of a focal point we have that $\rho$ will diverge, hence, the space-time becomes physically singular.  This proves part 1 of the theorem. As a side remark, notice that when $\gamma\leq1$, part 1 holds for any sign of $\rho_0$.
\\\\
(ii) Consider now a real scalar field $\psi$ characterized by a stress-energy tensor of the form 
\begin{equation}
{T}_{\alpha\beta}=\left(\nabla_{\alpha}\psi\right)\left(\nabla_{\beta}\psi\right)-\frac{1}{2}g_{\alpha\beta}\left(\nabla_{\mu}\psi\right)\left(\nabla^{\mu}\psi\right)-g_{\alpha\beta}V(\psi),\label{eq:SF_energy_momentum_tensor_general-1}
\end{equation}
where $V(\psi)$ represents a potential functional. Using the definitions in \eqref{eq:Energy_momentum_tensor_decomposition_quantities-1} and defining $\dot{\psi}:=v^{\alpha}\nabla_{\alpha}\psi$, we find
\begin{align}
\rho &
=\dot{\psi}^{2}+\frac{1}{2}\left(\nabla_{\mu}\psi\right)\left(\nabla^{\mu}\psi\right)+V\left(\psi\right)\,,\label{eq:SF_energy_density_general-1}\\
p &
=\frac{1}{N-1}\rho+\frac{N-2}{2\left(N-1\right)}\left(\nabla_{\mu}\psi\right)\left(\nabla^{\mu}\psi\right)-\frac{N}{N-1}V\left(\psi\right)\,,\label{eq:SF_pressure_general-1}\\
q_{\alpha} &
=q_{1\alpha}=q_{2\alpha}=-h_{~\alpha}^{\mu}\left(\nabla_{\mu}\psi\right)\dot{\psi}\,,\label{eq:SF_heat_flow_general-1}\\
\pi_{\alpha\beta} &
=\left[h_{~(\alpha}^{\mu}h_{\beta)}^{~\nu}-\frac{h_{\alpha\beta}}{N-1}h^{\mu\nu}\right]\left(\nabla_{\mu}\psi\right)\left(\nabla_{\nu}\psi\right)\,,\label{eq:SF_anisotropic_pressure_general-1}\\m_{\alpha\beta} & =0\,.
\end{align}
Let us now impose that the scalar field admits level surfaces such that its gradient, $\nabla^{\alpha}\psi$, is time-like. Since the choice of $v$ is arbitrary, we can impose without loss of generality that $v^{\alpha}$ is given by $\nabla^{\alpha}\psi$, so that, $h_{~\alpha}^{\beta}\nabla_{\beta}\psi=0$. This, then, implies that the quantities $q_\alpha$ and $\pi_{\alpha\beta}$ are null. Moreover,
\begin{align}
\rho &
=\frac{1}{2}\dot{\psi}^{2}+V\left(\psi\right)=\frac{1}{2}+V(\psi)\,,\label{eq:SF_energy_density_particular-1}\\
p &
=\frac{1}{2}\dot{\psi}^{2}-V\left(\psi\right)=\frac{1}{2}-V(\psi)\,,\label{eq:SF_pressure_particular-1}
\end{align}
where the last equalities follow from setting $v^{\alpha}:=\nabla^{\alpha}\psi$ and $v^{\alpha}v_{\alpha}=-1$.

From Eq.~\eqref{eq:conservation_timelike_perf_fluid}, and integrating along the fiducial curve of the congruence, we find
\begin{equation}
\rho\left(\tau\right)=\rho\left(\tau_{0}\right)+\left|\int_{\tau_{0}}^{\tau}\theta
dt\right|\,,\label{eq:SF_conservation_energy_integral_alt-1}
\end{equation}
where $\tau$ is the parameter along the curve's congruence. Assuming that the premises of the focus Theorem \ref{Theorem:timelike_focus_theorem} are verified, the function $\theta$ will diverge to negative infinite when $\tau$ goes to some finite value $\tau_{f}$. Then
\begin{equation}
\lim_{\tau\to\tau_{f}^{-}}\rho=+\infty\,,
\end{equation}
that is, the space-time becomes physically singular at $\tau=\tau_{f}$.
\end{proof} 

\subsection{Einstein-Cartan theory}

Theorem \ref{Theorem:singularities}, relating a focal point to  the formation of a singularity, is rather general in the sense that it is applicable to a wide class of theories of gravity. However, we had to specify the particular equation of state, either for a perfect fluid or a scalar field, in order to draw those results. One reason for this is that we have just used a general conservation law for $T$ and did not use the gravitational field equations of a particular gravity theory. Indeed, the field equations could enable us to directly relate the geometric conditions in the focus Theorem \ref{Theorem:timelike_focus_theorem} with constraints on the matter fields.

In this section, we will show that provided a geometric theory of gravity, it is possible to extend the results of Theorem \ref{Theorem:singularities} for perfect fluids without the need to impose an equation of state. In particular, we will consider the Einstein-Cartan theory defined in a $N=4$ dimensional Lorentzian manifold. The respective field equations read
\begin{align}
R_{\alpha\beta}-\frac{1}{2}g_{\alpha\beta}R &
+g_{\alpha\beta}\,\Lambda=8\pi
T_{\alpha\beta}\,,\label{eq:EC_field equations_1}\\
S^{\alpha\beta\gamma}+2g^{\gamma[\alpha}S^{\beta]}{}_{\mu}{}^{\mu}
& =-8\pi\Delta^{\alpha\beta\gamma}\,,\label{eq:EC_field
equations_2}
\end{align}
where $T_{\alpha\beta}$ represents the canonical stress-energy tensor, $\Delta^{\alpha\beta\mu}$ the intrinsic hypermomentum and $\Lambda$ the cosmological constant. Moreover, we have the following conservation laws \cite{obukhov,Hehl2}
\begin{equation}
\nabla_{\beta}T_{\alpha}{}^{\beta}=2S_{\alpha\mu\nu}T^{\nu\mu}-\frac{1}{4\pi}S_{\alpha\mu}{}^{\mu}\Lambda+\frac{1}{8\pi}\left(S_{\alpha\mu}{}^{\mu}R-S^{\mu\nu\sigma}R_{\alpha\sigma\mu\nu}\right)\,.\label{eq:EC_Conservation_equations}
\end{equation}
Computing the component $v_{\alpha}\nabla_{\beta}T^{\alpha\beta}$, we find
\begin{equation}
\begin{aligned}\nabla_{v}\rho+\left(\theta-W_{\alpha}^{\alpha}\right)\left(\rho+p\right)&+q_{1}^{\alpha}a_{\alpha}+\nabla_{\alpha}q_{2}^{\alpha}+\pi^{\alpha\beta}\left(\sigma_{\alpha\beta}-W_{\left(\alpha\beta\right)}\right)
+m^{\alpha\beta}\left(\omega_{\beta\alpha}-W_{\left[\beta\alpha\right]}\right)\\
& =-v^{\alpha}\left\{
2S_{\alpha\mu\nu}T^{\nu\mu}-\frac{1}{4\pi}S_{\alpha\mu}{}^{\mu}\Lambda+\frac{1}{8\pi}\left(S_{\alpha\mu}{}^{\mu}R-S^{\mu\nu\sigma}R_{\alpha\sigma\mu\nu}\right)\right\}
\,.
\end{aligned}
\label{eq:EC_energy_conservation_general}
\end{equation}
Now, from Eq.~\eqref{eq:EC_energy_conservation_general}, it is easy to verify that for a generic torsion tensor field we will have, on the right-hand side, terms that depend on the Weyl tensor, making it very difficult to infer, in general, the behavior of the relevant quantities without actually solving the field equations\footnote{This is also the case, even if we consider a torsion tensor field where $W_{\left(\alpha\beta\right)}=0$, as in Theorem \ref{Theorem:timelike_focus_theorem}.}. Nevertheless, we are able to prove the following result:
\begin{thm}
\label{Theorem:singularities_EC}
Consider a space-time $(M,g,S)$ in the context of the Einstein-Cartan theory of gravity, satisfying the conditions of Theorem \ref{Theorem:timelike_focus_theorem}. The space-time will
become physically singular at, or before, the formation of a focal point, if the following three conditions are verified:\leavevmode
\begin{enumerate}
\item the torsion tensor field can be written as $S_{\alpha\beta}{}^{\gamma}=S_{\alpha\beta}v^{\gamma}+v_{[\alpha}A_{\beta]}v^{\gamma}$;
\item the matter fluid is characterized by a canonical stress-energy tensor of the form $T_{\alpha\beta}=\rho v_{\alpha}v_{\beta}+p\,h_{\alpha\beta}$, and at some instant, $\tau=\tau_0$ , $\theta(\tau_0)<0 \wedge \rho(\tau_0)\geq0$;
\item the cosmological constant, $\Lambda$, verifies $\Lambda + S_{\alpha\beta}S^{\alpha\beta} \geq 0$.
\end{enumerate}
\end{thm}

\begin{proof}
We start by noting that, using condition 1, where the components $S_{\alpha\beta}$ and $A_{\alpha}$ were defined in Eq.~\eqref{eq:timelike_torsion_decomposition_components_def}, Eq.~\eqref{eq:EC_energy_conservation_general} reads
\begin{equation}
\label{new-rel}
\nabla_{v}\rho+\theta\left(\rho+p\right)+q_{1}^{\alpha}a_{\alpha}+q_{2}^{\alpha}A_{\alpha}+\nabla_{\alpha}q_{2}^{\alpha}+\pi^{\alpha\beta}\sigma_{\alpha\beta}+m^{\alpha\beta}\omega_{\beta\alpha}=0,
\end{equation}
and using condition 2, this equation reduces simply to Eq.~\eqref{eq:conservation_timelike_perf_fluid}.  Now, consider $h$ as the induced metric on each $N-1$ hypersurface, orthogonal to $v$, with components $h_{ab}$ in a coordinate basis. It is easy to show that the definition \eqref{eq:timelike_expansion_def} for the expansion scalar is equivalent to
\begin{equation}
\theta=h^{ab}\mathcal{L}_{v}h_{ab}=\frac{1}{\sqrt{\det\left(h\right)}}\partial_{v}\sqrt{\det\left(h\right)}\,.
\end{equation}
Defining $L:=\left(\sqrt{\det\left(h\right)}\right)^{\frac{1}{3}}$ along the integral curve of $v$, we can rewrite Eq.~\eqref{eq:conservation_timelike_perf_fluid} as
\begin{equation}
\frac{d}{dL}\left(\rho
L^{2}\right)+L\left(\rho+3p\right)=0\,.\label{eq:EC_conservation_energy_alt}
\end{equation}
The second condition in Theorem \ref{Theorem:timelike_focus_theorem}, $R_{\alpha\beta}v^{\alpha}v^{\beta}\geq\left(W_{\left[\alpha\beta\right]}+S_{\alpha\beta}\right)S^{\alpha\beta}$, in the considered Einstein-Cartan setup, implies
\begin{equation}
4\pi\left(3p+\rho\right)-\Lambda\geq
S_{\alpha\beta}S^{\alpha\beta}\,.\label{eq:EC_generalized_SEC}
\end{equation}
Integrating Eq.~\eqref{eq:EC_conservation_energy_alt} along the integral curve of $v$ yields
\begin{equation}
\rho=\frac{\left.\left(\rho
L^{2}\right)\right|_{\tau_{0}}}{L^{2}}+\frac{1}{L^2}\int_{L}^{L_{0}}l
(\rho+3p)\,dl\,,
\end{equation}
where $\tau_{0}$ represents the value of the affine parameter of the integral curve of $v$, defined so that $L\left(\tau_{0}\right)=L_{0}>0$. Taking the limit when $L\to0$, i.e. at the formation of a focal point, and assuming $\Lambda + S_{\alpha\beta}S^{\alpha\beta} \geq 0$, then from Eq.~\eqref{eq:EC_generalized_SEC}, we find that the second term in the previous equation is positive and, therefore, $\rho \to +\infty$.
\end{proof}

\section{Conclusions and further discussion}

Our results show how the presence of torsion affects the formation of focal points of a congruence of curves on a Lorentzian manifold. In particular, we found how torsion explicitly appears in the conditions of the newly derived focus theorems which, in turn, can be used to study singularity formation. 
Using the focus theorems, we extended the Raychaudhuri-Komar singularity theorem to space-times endowed with a non-null torsion tensor field, for a wide class of affine theories of gravity. Moreover, we discussed both the cases of metric geodesics and accelerated curves. This is an important point from the physical point of view. Take, for instance, the somewhat simpler case of the theory of General Relativity. In that case, single-pole massive particles will follow geodesics of the space-time \cite{Papapetrou}. However, this is only verified if no gravitational radiation due to the particle's motion is taken into consideration \cite{Papapetrou, Poisson_Pound}. Moreover, not only gravitational radiation will spoil the geodesic motion of single-pole particles but also charged particles in curved space-times will be accelerated, as described by the deWitt-Brehme-Hobbs equation \cite{DeWitt, Hobbs}. Regarding the motion of massive particles in a space-time with non-vanishing torsion, the situation is even more complex. As shown by Puetzfeld and Obukhov \cite{obukhov}, for an affine gravity theory minimally coupled to the matter fields, only single-pole massive particles with null intrinsic hypermomentum will follow geodesics of the space-time. However, this result was found by assuming that no radiation is emitted. It is then expected that, analogously to the case of space-times with null torsion, self-force effects will give rise to non-zero acceleration. The results in this article address these cases by showing how the presence of acceleration and torsion affects the formation of focal points of a congruence. 

\section*{Acknowledgments}

We are grateful to José Natário, Miguel Sanchez and José Senovilla for useful discussions. PL thanks IDPASC and FCT-Portugal for financial support through Grant No. PD/BD/114074/2015. FM thanks: FCT project PTDC/MAT-ANA/1275/2014; CAMGSD, IST, Univ. Lisboa, through FCT project UID/MAT/04459/2013; CMAT, Univ. Minho, through FCT project Est-OE/MAT/UI0013/2014 and FEDER Funds COMPETE.

\end{document}